\documentclass[final,5p,times]{elsarticle}

\usepackage{amsmath,amssymb,mathtools}
\usepackage{amsthm}
\usepackage{microtype}
\usepackage{tikz}
\usetikzlibrary{automata,positioning,arrows.meta,decorations.pathreplacing}
\usepackage{placeins}
\usepackage[hidelinks]{hyperref}
\hypersetup{pdftitle={Further Remarks on Separating Words}}

\newtheorem{theorem}{Theorem}
\newtheorem{proposition}[theorem]{Proposition}
\newtheorem{lemma}[theorem]{Lemma}

\newcommand{\sepw}{\operatorname{sep}}
\newcommand{\nsep}{\operatorname{nsep}}
\newcommand{\rev}[1]{#1^{R}}
\newcommand{\shift}[1]{\rho_{#1}}

\begin{document}

 \begin{frontmatter}

 \title{Further Remarks on Separating Words}

 \author{John W. Nicol\corref{cor1}}
 \ead{jnicol1994@alumni.cmu.edu}
 \cortext[cor1]{Corresponding author}
 
 \begin{abstract}
 We revisit questions on separating words raised by Demaine, Eisenstat,
 Shallit, and Wilson, together with Ebrahimnejad's follow-up to their reversal
 problem.
 For length-$n$ pairs whose difference word has $d$ runs, we prove an
 $O(d\log n)$ bound, extending the Hamming-distance theorem of Demaine et al.
 For conjugate words, we give bounds controlled by the arithmetic of the shift.
 We resolve Demaine et al.'s Open Problem~2 by showing that the order of two
 words can change nondeterministic separation by an unbounded factor.
 Our reversal construction addresses Ebrahimnejad's follow-up to Open Problem~1:
 forward and reversed deterministic separation can differ by an unbounded
 factor.
 Since nondeterministic separation is invariant under reversal, the same
 construction also improves the lower bound in Open Problem~3.
 \end{abstract}

 \end{frontmatter}

 \section{Introduction}

 A deterministic finite automaton (DFA) separates two distinct words if it accepts one and rejects the other.

 The problem was introduced by Goral\v{c}\'{i}k and Koubek~\cite{GoralcikKoubek1986}.
 Robson obtained the long-standing $O(n^{2/5}(\log n)^{3/5})$-state upper bound and also studied separation by automata whose transition monoids are groups~\cite{Robson1989,Robson1996}.
 A 2011 paper of Demaine, Eisenstat, Shallit, and Wilson surveyed the subject and posed questions for several restrictions and variants~\cite{DemaineEtAl2011}.
 Subsequent lower-bound work developed the connection with identities in finite transformation semigroups and symmetric groups~\cite{BulatovEtAl2017}.
 More recently, Chase~\cite{Chase2021} improved the upper bound to $O(n^{1/3}(\log n)^7)$ states.

 We organize this paper around the questions in~\cite{DemaineEtAl2011}.
 Sections~2 and~3 revisit its results and questions for words with additional
 structure: runs of differences and cyclic shifts.
 Sections~4 and~5 address its three open problems, together with
 Ebrahimnejad's follow-up question to Open Problem~1~\cite{Ebrahimnejad2018}.

 Theorem 2 of~\cite{DemaineEtAl2011} shows that $h$ differing
 positions (a Hamming distance of $h$) can be detected with $O(h\log n)$ states.
 We generalize this to detecting $d$ runs of differences in $O(d\log n)$ states,
 regardless of how many positions differ.
 This also addresses a remark in~\cite{Chase2021}.

 Section~5.2 of~\cite{DemaineEtAl2011} asks whether conjugate words, which are
 related by a cyclic shift, are easier to separate.
 We show that a shift with large greatest common divisor contains an arbitrary shorter
 separating-words instance, while a shift with small greatest common divisor can be handled by a small
 fingerprint.

 Open Problem~2 asks whether swapping the accepted and rejected words can
 change nondeterministic separation by an unbounded amount.
 We show that both the difference and the ratio can be arbitrarily large, even
 for equal-length binary words that differ in only two positions.

 Open Problem~1 asks whether reversal can change deterministic separation by
 an unbounded amount.
 Ebrahimnejad proved that the difference is unbounded and asked whether there
 is a good upper bound on the corresponding ratio~\cite{Ebrahimnejad2018}.
 We rule out a constant upper bound by constructing pairs for which forward
 separation needs more than $N$ states while reversed separation needs only
 $O(\log N)$ states.

 Open Problem~3 asks for better bounds on the ratio between deterministic and
 nondeterministic separation as a function of the common word length.
 Since nondeterministic separation is unchanged by reversal, the same
 construction improves its lower bound from
 $\Omega(\sqrt{\log L})$ to
 $\Omega(\log L/(\log\log L)^2)$ for every sufficiently large $L$.

 We state the reversal theorem in Section~4.2, where it is used for
 nondeterministic separation, and defer its proof to the final section.
 We write $\sepw(w,x)$ for the number of states in the smallest DFA separating words $w$ and $x$\@.
 Unless stated otherwise, all words are binary and all DFAs are complete.
 All logarithms are to base $2$.

 \section{Hamming distance and runs of differences}

 For equal-length binary words $w,x\in\{0,1\}^n$, let
 $\Delta=w\oplus x$ be their difference word.
 Here $\Delta_i=1$ exactly where $w_i\ne x_i$.
 A \emph{run of differences} is a maximal run of $1$'s in $\Delta$.
 Put $\Delta_0=\Delta_{n+1}=0$, and call $i\in\{1,\ldots,n+1\}$ a
 \emph{boundary} if $\Delta_{i-1}\ne\Delta_i$.
 A difference word with $d$ runs has exactly $2d$ boundaries.

 The proof applies the Hamming-distance argument from Theorem~2 of~\cite{DemaineEtAl2011} to the boundaries instead of all differing positions.

 \begin{theorem}
 \label{thm:few-runs}
 If $w\ne x$ and their difference word has $d$ runs, then
 \[
 \sepw(w,x)=O(d\log n).
 \]
 \end{theorem}

 \begin{proof}
 Let $B$ be the set of boundaries, choose $b\in B$, and put
 \[
 M=\prod_{i\in B\setminus\{b\}}|i-b|\le(n+1)^{2d-1}.
 \]
 By the prime number theorem, the product of the primes up to $C d\log n$ exceeds $M$ for a suitable constant $C$.
 Thus some prime $p=O(d\log n)$ does not divide $M$, making $b$ the only boundary in its residue class modulo $p$.

 For a binary word $u$ and a residue $c$ modulo $p$, define the parity fingerprint
 \[
 a_{p,c}(u)=\sum_{j\equiv c\pmod p}u_j\pmod 2.
 \]
 Each $1$ at position $j$ in $\Delta$ contributes to the two possible boundaries $j$ and $j+1$; adjacent contributions cancel modulo $2$.
 The parity of the boundaries congruent to $b$ is
 \[
 a_{p,b-1}(\Delta)\oplus a_{p,b}(\Delta).
 \]
 This parity is $1$, so $a_{p,c}(\Delta)=1$ for $c=b-1$ or $c=b$.
 Since $\Delta=w\oplus x$, $a_{p,c}(w)\ne a_{p,c}(x)$.
 As in that argument, a DFA with two rings of $p$ states computes $a_{p,c}$: it records the input position modulo $p$ and switches rings whenever it reads a $1$ in residue class $c$.
 Choosing the appropriate ring as accepting separates $w$ and $x$ with $2p=O(d\log n)$ states.
 \end{proof}

 The Hamming-distance result follows: if $w$ and $x$ differ in $h$ positions, then their difference word has at most $h$ runs, so $\sepw(w,x)=O(h\log n)$.

 Chase~\cite[Section~2]{Chase2021} mentioned the Thue--Morse word and its complement, each with padding, as self-similar examples that were not then known how to handle for
 separating words, trace reconstruction, or reconstruction from subsequences.
 Since those words differ by one complemented middle interval, $d=1$, and Theorem~\ref{thm:few-runs} gives an $O(\log n)$-state separator.

 \section{Conjugate words}

 Two words are conjugates if one is a cyclic shift of the other.
 For a length-$n$ word $W=W_0W_1\cdots W_{n-1}$, let $\shift{k}(W)$ be its left shift by $k$ positions:
 \[
 \shift{k}(W)_i=W_{(i+k)\bmod n}.
 \]
 For $n\ge1$, put
 \[
 S(n)=\max_{w\ne x,\ |w|=|x|=n}\sepw(w,x).
 \]
 This is the exact-length convention; some earlier work instead maximizes over pairs whose lengths are at most $n$.
 For $0<k<n$, define
 \[
 R(n,k)=
 \max_{\substack{|W|=n\\W\ne\shift{k}(W)}}
 \sepw(W,\shift{k}(W)).
 \]
 Put $g=\gcd(n,k)$.

 \begin{theorem}
 \label{thm:conjugates}
 For every $0<k<n$,
 \[
 \begin{aligned}
 S(g)&\le R(n,k),\\
 R(n,k)&=O\bigl(\min\{g\log^2(n/g),k,n-k\}\bigr).
 \end{aligned}
 \]
 \end{theorem}

 The lower bound embeds the ordinary separating-words problem at length $g$,
 while the upper bounds give $O(\log^2 n)$ states when $g=1$ and $O(1)$ states when $k$ or $n-k$ is fixed.

 Combining the three propositions below proves the theorem.

 \begin{proposition}
 \label{prop:conjugate-lower}
 $S(g)\le R(n,k)$.
 \end{proposition}

 \begin{proof}
 Write $r=k/g$ and $s=(n-k)/g$, so $r+s=n/g$.
 Take any distinct binary words $x,y$ of length $g$, and form
 \[
 W=x^r y^s.
 \]
 Shifting by $k=rg$ gives $\shift{k}(W)=y^s x^r$.
 These words are distinct because their first length-$g$ blocks are $x$ and $y$.

 If a DFA separates these two words, the transformations induced by $x$ and $y$ differ on some state $q$; otherwise the two long products would induce the same transformation.
 Starting the same transition graph at $q$ and choosing one of these two target states to be accepting gives a DFA of the same size separating $x$ from $y$.
 Maximizing over $x,y$ proves the claim.
 \end{proof}

 \begin{proposition}
 \label{prop:conjugate-gcd-upper}
 $R(n,k)=O\bigl(g\log^2(n/g)\bigr)$.
 \end{proposition}

 \begin{proof}
 Fix $W\ne\shift{k}(W)$, and put
 \[
 r=\frac{k}{g},
 \qquad
 \ell=\frac{n-k}{g}.
 \]
 Consider the positions in one fixed residue class modulo $g$, which form a binary word of length $n/g$.
 The shift moves each symbol backward by $r$ positions modulo $n/g$.
 Since $W\ne\shift{k}(W)$, at least one of these $g$ words is nonconstant; fix such a word.

 Let $a$ be the number of its $1$'s that move backward without wrapping, and let $b$ be the number that wrap around.
 The former move by $-r$ positions, while the latter move by $\ell$ positions.
 The sum of the positions of the $1$'s changes by
 \[
 \Delta=\ell b-ra.
 \]
 Suppose $\Delta=0$.
 Then $\gcd(r,\ell)=1$ implies that $\ell\mid a$ and $r\mid b$.
 Since there are only $\ell$ nonwrapping positions and $r$ wrapping positions, this would make the chosen word constant, contradicting its choice.

 Since $0<|\Delta|<(n/g)^2$, some prime $p=O\bigl(\log(n/g)\bigr)$ does not divide $\Delta$.
 Modulo $p$, the sum of the positions of the $1$'s is determined by the numbers of $1$'s in the $p$ residue classes of positions.
 Thus one residue class has different counts modulo $p$ before and after the shift.

 In the full words, this is one residue class modulo $gp$.
 A DFA tracks the input position modulo $gp$ and counts modulo $p$ the $1$'s in that class.
 It has $gp^2=O\bigl(g\log^2(n/g)\bigr)$ states and separates $W$ and $\shift{k}(W)$.
 \end{proof}

 \begin{proposition}
 \label{prop:small-shift}
 $R(n,k)=O(\min\{k,n-k\})$.
 \end{proposition}
 \begin{proof}
  Since $\shift{n-k}$ is the inverse of $\shift{k}$ and $\sepw$ is symmetric,
  \[
   R(n,k)=R(n,n-k).
  \]
  It is enough to prove the explicit bound $R(n,k)\le 2k+1$.

  Fix $W\ne\shift{k}(W)$, and write
  \[
   V=\shift{k}(W).
  \]
  Let $U$ be the infinite $k$-periodic word obtained by repeating the first
  $k$ symbols of $W$.

  Suppose that $W$ first differs from $U$ at position $L$.
  Then $L\ge k$, and for every $i<L-k$,
  \[
   V_i=W_{i+k}=U_{i+k}=U_i,
  \]
  whereas
  \[
   V_{L-k}=W_L\ne U_L=U_{L-k}.
  \]
  The first difference between $V$ and $U$ occurs at position $L-k$.
  If $W$ never differs from $U$, then $V$ must differ from $U$, since
  otherwise both words would equal the length-$n$ prefix of $U$.

  A DFA can follow $U$ using $k$ states to track the input position modulo
  $k$.  At the first difference, it switches to a cycle of $k+1$ states that
  counts the remaining input length modulo $k+1$.
  If only one of $W$ and $V$ differs from $U$, one finishes among the first
  $k$ states and the other on the $(k+1)$-cycle.
  If both differ, their first differing positions are $k$ apart, so they
  finish at different positions on that cycle.
  A DFA with $2k+1$ states therefore separates $W$ and $V$.
  The symmetry above gives
  \[
   R(n,k)\le 2\min\{k,n-k\}+1
   =O(\min\{k,n-k\}).
  \]
 \end{proof}

 \section{Nondeterministic separation}

 An NFA has one initial state, an arbitrary set of accepting states, and a
 possibly partial transition relation.
 It accepts a word if some path labeled by that word leads from the initial
 state to an accepting state.
 We do not use $\varepsilon$-transitions; they can be eliminated without
 adding states.
 We write $\nsep(w,x)$ for the number of states in the smallest NFA accepting
 $w$ and rejecting $x$; the order of the two words matters.

 \subsection{The order of the two words}

 Open Problem~2 of~\cite{DemaineEtAl2011} asks whether the difference
 \[
  |\nsep(w,x)-\nsep(x,w)|
 \]
 between the two orientations can be unbounded.
 We show that it can; moreover, the corresponding ratio can be unbounded,
 even when the two words differ in only two positions.

 The proof adapts the consecutive-cycle construction from Theorem~3 of~\cite{DemaineEtAl2011}.

 \begin{theorem}
  \label{thm:nsep-asymmetry}
  For every integer $s\ge2$, there are equal-length binary words $u_s,v_s$ such that
  \[
   \nsep(u_s,v_s)>s^2,
   \qquad
   \nsep(v_s,u_s)\le s+1.
  \]
  Both the difference and the ratio between the two orientations are unbounded.
 \end{theorem}

 \begin{proof}
  Put
  \[
   M=s^2,
   \qquad
   T=2^M,
   \qquad
   P=\operatorname{lcm}(1,2,\ldots,T),
  \]
  and define
  \[
   u_s=0^M1\,0^{T+P},
   \qquad
   v_s=0^{M+P}1\,0^T.
  \]
  These words have the same length and differ in two positions.

  Suppose an NFA with $m\le M$ states accepts $u_s$.
  Some accepting path has the form
  \[
   i\xrightarrow{0^M}p\xrightarrow{1}q\xrightarrow{0^{T+P}}f.
  \]
  Since $M\ge m$, the first segment contains a loop of some length $d\le m$.
  Since $d\mid P$, repeating this loop $P/d$ additional times gives a path
  from $i$ to $p$ labeled $0^{M+P}$.

  On the final segment, let $S_j$ be the set of states reachable from $q$
  by $0^j$.
  There are at most $2^m\le T$ possible sets $S_j$, and $S_{j+1}$ is
  determined by $S_j$.
  Thus $S_T$ lies on a cycle of some length $c\le T$.
  Since $c\mid P$, it reaches the same subset after $T+P$ zeros.
  Hence $f$ is reachable from $q$ by $0^T$ as well.
  The resulting path accepts $v_s$.
  Hence
  \[
   \nsep(u_s,v_s)>M=s^2.
  \]

  For the opposite orientation, use the $(s+1)$-state NFA in
  Figure~\ref{fig:asymmetry-nfa}.

  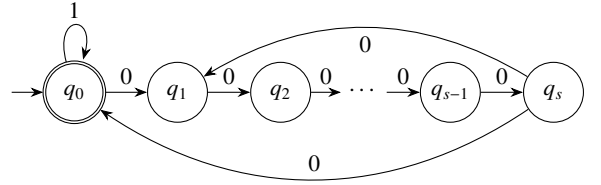
\begin{figure}[htbp]
   \centering
   \begin{tikzpicture}
    [
    >=Stealth,
    every state/.style={minimum size=9mm, inner sep=1pt},
    every edge/.style={draw,->},
    initial text={},
    scale=.88,
    transform shape
    ]
    \node[state, initial, accepting] (q0) at (0,0) {$q_0$};
    \node[state] (q1) at (1.55,0) {$q_1$};
    \node[state] (q2) at (3.1,0) {$q_2$};
    \node (dots) at (4.35,0) {$\cdots$};
    \node[state] (qsm1) at (5.65,0) {$q_{s-1}$};
    \node[state] (qs) at (7.2,0) {$q_s$};

    \path
 (q0) edge[loop above] node {$1$} (q0)
 (q0) edge node[above] {$0$} (q1)
    (q1) edge node[above] {$0$} (q2)
    (q2) edge node[above] {$0$} (dots)
    (dots) edge node[above] {$0$} (qsm1)
    (qsm1) edge node[above] {$0$} (qs)
    (qs) edge[bend left=34] node[above] {$0$} (q0)
    (qs) edge[bend right=31] node[below] {$0$} (q1);
   \end{tikzpicture}
   \caption{The $(s+1)$-state NFA separating $v_s$ from $u_s$.}
   \label{fig:asymmetry-nfa}
  \end{figure}
  \FloatBarrier

  For $\ell>0$, the NFA accepts $0^\ell$ exactly when
  \[
   \ell=as+b(s+1),\qquad a\ge0,\quad b\ge1.
  \]
  Every return to $q_0$ uses the $(s+1)$-cycle once, while before
  such a return the path may traverse the $s$-cycle any number of times.
  Modulo $s$, we have $\ell\equiv b$.
  Taking $1,\ldots,s$ as representatives modulo $s$, the smallest possible
  $b\ge1$ in residue class $r$ is $r$; taking $a=0$ gives the least accepted
  length $r(s+1)$ in that class, and increasing $a$ gives every later length.
  Since lengths in a fixed residue class differ by $s$, the largest
  nonaccepted length in residue class $r$ is $r(s+1)-s$.
  Thus the largest $\ell$ for which $0^\ell$ does not take $q_0$
  back to $q_0$ is
  \[
   \max_{1\le r\le s}\bigl(r(s+1)-s\bigr)
   =s(s+1)-s=s^2=M.
  \]
  Since $P>0$ and $T=2^M>M$, both blocks of zeros in $v_s$ are longer than $M$.
  Hence the NFA accepts $v_s$.
  No path labeled $0^M$ returns from $q_0$ to $q_0$.
  Since $q_0$ is the only state with a transition labeled $1$, the NFA
  rejects $u_s$.
  Thus
  \[
   \nsep(v_s,u_s)\le s+1.
  \]
 \end{proof}

 Define
 \[
 A(q)=\max\left\{\nsep(w,x):
 \substack{|w|=|x|,\ w\ne x,\\ \nsep(x,w)\le q}
 \right\}.
 \]
 For $q\ge3$, Theorem~\ref{thm:nsep-asymmetry} with $s=q-1$ gives the lower bound,
 while determinizing and complementing gives the upper bound:
 \[
 (q-1)^2+1\le A(q)\le2^q.
 \]
 The exact growth of $A(q)$ remains open.

 \subsection{Deterministic versus nondeterministic separation}
 Open Problem~3 of~\cite{DemaineEtAl2011} asks for better bounds on the ratio
 of deterministic to nondeterministic separation as a function of word length.
 For the exact common-length convention, define
 \[
 H(L)=\max\left\{\frac{\sepw(w,x)}{\nsep(w,x)}:
 |w|=|x|=L,\ w\ne x\right\}.
 \]

 For a word $w$, write $\rev{w}$ for its reversal.

 We first use the following reversal theorem proved in Section~\ref{sec:reversals} to improve the lower bound on
 $H(L)$, and then use the NFA size itself as the parameter.

 \begin{theorem}
  \label{thm:reversal-gap}
  There is an absolute constant $c_0>0$ such that, for every $N\ge2$ and every
  integer $L\ge N^{c_0N}$, there are distinct words
  $x_{N,L},y_{N,L}\in\{0,1\}^L$ such that
  \[
   \begin{aligned}
    \sepw(x_{N,L},y_{N,L})&>N,\\
    \sepw(\rev{x_{N,L}},\rev{y_{N,L}})&=O(\log N).
   \end{aligned}
  \]
 \end{theorem}

 It was observed in~\cite{DemaineEtAl2011} that nondeterministic separation is
 unchanged by reversal:
 \[
 \nsep(w,x)=\nsep(\rev{w},\rev{x}).
 \]
 To see this, take an NFA accepting $w$ and rejecting $x$, and fix an accepting
 path on $w$ ending at a state $f$.
 Reverse every transition, use $f$ as the initial state, and make the old
 initial state the sole accepting state.
 This NFA accepts $\rev{w}$, while an accepting path on $\rev{x}$ would reverse
 to a path on $x$ ending at $f$, contradicting rejection of $x$.
 The construction preserves the state set, and applying it in the other
 direction proves equality.

 Apply this observation to the words in Theorem~\ref{thm:reversal-gap}.
 The small DFA for the reversed pair gives nondeterministic separation
 $O(\log N)$.
 In the other direction, determinizing any separating $q$-state NFA gives at
 most $2^q$ states.
 Since the deterministic separation is greater than $N$, the nondeterministic
 separation is $\Theta(\log N)$.

 For each sufficiently large $L$, choose $N$ to be a sufficiently small
 constant multiple of
 \[
 \frac{\log L}{\log\log L}.
 \]
 Then $L\ge N^{c_0N}$, where $c_0$ is the constant in
 Theorem~\ref{thm:reversal-gap}.
 The resulting length-$L$ pair gives
 \[
 H(L)=\Omega\!\left(\frac{\log L}{(\log\log L)^2}\right).
 \]
 This improves the $\Omega(\sqrt{\log L})$ lower bound
 in~\cite{DemaineEtAl2011} and holds for every sufficiently large $L$.

 The same construction also shows that the exponential cost of
 determinization is unavoidable.
 For the resulting pairs, let $q=\nsep(w,x)$.
 The reversed separator and invariance under reversal give
 $q=O(\log N)$, while $\sepw(w,x)>N$.
 Together with the determinization bound $\sepw(w,x)\le2^q$, this gives
 \[
  \sepw(w,x)=2^{\Theta(q)}
  =2^{\Theta(\nsep(w,x))}.
 \]

 \section{Reversals}
 \label{sec:reversals}

 This section proves Theorem~\ref{thm:reversal-gap}.

 Like Ebrahimnejad's construction~\cite{Ebrahimnejad2018}, ours uses a
 language $C_k$ whose reversal has a small DFA, together with a
 diagonalization over all DFAs with at most $N$ states.
 The diagonalization produces equal-length words $u_N\in C_k$ and
 $v_N\notin C_k$ that have the same effect on every state of every such DFA\@.
 Thus the forward pair requires more than $N$ states to separate, whereas the
 reversed pair has an $O(\log N)$-state separator.
 Finally, we encode the pair in binary and pad it to every sufficiently large
 length.

 \subsection{The delimiter language}

 For $k\ge1$, let $\Gamma=\{1,2,\#\}$ and put
 \[
 G_k=\bigl(\{1,2\}^*1\{1,2\}^{k-1}\#\bigr)^+,
 \qquad
 \overline{G_k}=\Gamma^*\setminus G_k.
 \]
 Thus $G_k$ consists of one or more blocks ending in $\#$, with $1$ as the
 $k$th symbol before each $\#$.
 Over $\Sigma=\Gamma\cup\{@\}$, define
 \[
 C_k=(@\overline{G_k})^*.
 \]
 The two delimiters have different roles: $\#$ ends the inner blocks, while
 $@$ separates them.
 The two delimiters make the decomposition unique: every word in $C_k$
 is a concatenation of blocks $@h$ with
 $h\in\overline{G_k}$.
 In particular, $@z@\notin C_k$ for $z\in G_k$, whereas
 $@h@\in C_k$ for $h\in\overline{G_k}$.

 Reversing the definition gives
 \[
 \rev{C_k}
 =\bigl((\Gamma^*\setminus\rev{G_k})@\bigr)^*,
 \qquad
 \rev{G_k}
 =\bigl(\#\{1,2\}^{k-1}1\{1,2\}^*\bigr)^+.
 \]
 \begin{lemma}[A small automaton for the delimiter language]
 \label{lem:delimiter-automata}
 The language $\rev{C_k}$ has an $O(k)$-state DFA.
 \end{lemma}

 \begin{proof}
  Complete and complement the DFA in
  Figure~\ref{fig:reversed-delimiter-dfa} to recognize
  $\Gamma^*\setminus\rev{G_k}$.
  Since $@\notin\Gamma$, recognizing
  $\bigl((\Gamma^*\setminus\rev{G_k})@\bigr)^*$ requires only a constant
  number of additional states: check each $\Gamma$-word when its terminating
  $@$ is read, and restart after each successful check.
  A new initial accepting state is used both at the beginning of the input
  and after each successful check, while a rejecting sink handles failed
  checks.
  Thus an $@$ returns to the new initial state exactly when the preceding
  $\Gamma$-word is outside $\rev{G_k}$, and otherwise enters the rejecting sink.
  All other states are nonaccepting.
  Thus $\rev{C_k}$ has an $O(k)$-state DFA.
 \end{proof}

 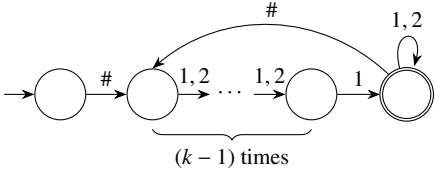
\begin{figure}[htbp]
  \centering
  \begin{tikzpicture}
   [
   >=Stealth,
   every state/.style={minimum size=8mm, inner sep=1pt},
   every edge/.style={draw,->},
   initial text={},
   scale=.84,
   transform shape
   ]
   \node[state, initial] (B) at (0,0) {};
   \node[state] (q1) at (1.45,0) {};
   \node (dots) at (2.7,0) {$\cdots$};
   \node[state] (qk) at (3.95,0) {};
   \node[state, accepting] (G) at (5.45,0) {};

   \path
   (B) edge node[above] {$\#$} (q1)
   (q1) edge node[above] {$1,2$} (dots)
   (dots) edge node[above] {$1,2$} (qk)
   (qk) edge node[above] {$1$} (G)
   (G) edge[loop above] node {$1,2$} (G)
       edge[bend right=42] node[above] {$\#$} (q1.north);
   \draw[decorate, decoration={brace,mirror,amplitude=4pt}]
   (1.45,-.55) -- (3.95,-.55)
   node[midway,below=5pt] {$(k-1)$ times};
  \end{tikzpicture}
  \caption{A partial DFA for $\rev{G_k}$.
  When $k=1$, the part under the brace consists of a single state.
  All undisplayed transitions are undefined.}
  \label{fig:reversed-delimiter-dfa}
 \end{figure}

 \subsection{Replacement and a common suffix}

 Fix $N\ge2$, and let $k$ be the least integer such that $2^k>N^2$.
 Thus $k=O(\log N)$.
 Put $Q=\{1,\ldots,N\}$, and let $\mathcal D_N$ be the set of all complete
 transition tables
 \[
 \delta_D:Q\times\Sigma\longrightarrow Q.
 \]
 We use the same notation for the extension of $\delta_D$ to words and to sets
 of states.
 The set $\mathcal D_N$ has size $N^{|\Sigma| N} = N^{O(N)}$.
 Any DFA with at most $N$ states can be padded with unreachable states, so it
 is enough to make the construction work for every $D\in\mathcal D_N$.

 \begin{lemma}[Universal replacement word]
 \label{lem:replacement-word}
 There is a nonempty word $z\in G_k$ of length $N^{O(N)}$ such that, for every
 $D\in\mathcal D_N$ and $p,q\in Q$, some $h\in\overline{G_k}$ with $|h| = |z|$ satisfies
 \[
 \begin{aligned}
 \delta_D(p,z)&=\delta_D(p,h),\\
 \delta_D(q,z)&=\delta_D(q,h).
 \end{aligned}
 \]
 \end{lemma}

 \begin{proof}
 Enumerate the $m=|\mathcal D_N|N^2$ triples as
 \[
 (D_1,p_1,q_1),\ldots,(D_m,p_m,q_m).
 \]
 We recursively construct pairs $x_i,y_i$ of equal length and put
 $z_i=x_1\cdots x_i$, starting with $z_0=\varepsilon$.
 At stage $i$, the word $x_i$ will belong to $G_k$, while $y_i$
 will not, and replacing the final block $x_i$ in
 $z_i=z_{i-1}x_i$ by $y_i$ will preserve the action of the whole
 prefix on $p_i$ and $q_i$ in $D_i$.

 Suppose that $z_{i-1}$ has been constructed.
 Consider the $2^k$ pairs
 \[
 \bigl(\delta_{D_i}(p_i,z_{i-1}s),
       \delta_{D_i}(q_i,z_{i-1}s)\bigr),
 \qquad
 s\in\{1,2\}^k.
 \]
 There are only $N^2$ possible pairs, so two distinct words $s_i,t_i$ give the
 same pair.
 Choose a position $\ell_i$, numbered from the left, where they differ,
 interchanging them if necessary so that
 $(s_i)_{\ell_i}=1$ and $(t_i)_{\ell_i}=2$.
 Put
 \[
 x_i=s_i2^{\ell_i-1}\#,
 \qquad
 y_i=t_i2^{\ell_i-1}\#.
 \]
 Their $k$th symbols before the sole $\#$ are $1$ and $2$, respectively.
 Hence $x_i\in G_k$, $y_i\in\overline{G_k}$, and
 $|x_i|=|y_i|=O(k)$.
 Reading the common suffix $2^{\ell_i-1}\#$ preserves equality of the state
 pairs, so, with $z_i=z_{i-1}x_i$,
 \[
 \begin{aligned}
 \delta_{D_i}(p_i,z_i)
 &=\delta_{D_i}(p_i,z_{i-1}y_i),\\
 \delta_{D_i}(q_i,z_i)
 &=\delta_{D_i}(q_i,z_{i-1}y_i).
 \end{aligned}
 \]

 At the end, put $z=z_m$.
 For the triple $(D_j,p_j,q_j)$, take
 \[
 h=x_1\cdots x_{j-1}y_jx_{j+1}\cdots x_m.
 \]
 The $\#$ symbols force this block decomposition, so
 $z\in G_k$ and $h\in\overline{G_k}$, with $|h| = |z|$.
 The equalities obtained at stage $j$ persist after appending the common suffix
 $x_{j+1}\cdots x_m$, giving the required equalities for $z$ and $h$.
 Finally, there are $m=N^{O(N)}$ stages, each adding $O(k)$ symbols, so
 $|z|=N^{O(N)}$.
 \end{proof}

 Fix the word $z$ from Lemma~\ref{lem:replacement-word}, and put
 \[
 \mathcal R_z
 =\{\,@h:h\in\overline{G_k},\ |h|=|z|\,\}.
 \]
 Then $\mathcal R_z^*\subseteq C_k$.
 In particular,
 \[
 @1^{|z|}\in\mathcal R_z,
 \]
 because $1^{|z|}$ has no $\#$, so
 $1^{|z|}\in\overline{G_k}$.
 Every replacement $h$ supplied by Lemma~\ref{lem:replacement-word} satisfies
 $@h\in\mathcal R_z$.

 \begin{lemma}[A common suffix]
 \label{lem:stable-suffix}
 There is a word $w\in\mathcal R_z^*$, beginning with $@$ and of length
 $N^{O(N)}$, such that, for every $D\in\mathcal D_N$ and every
 $c\in\mathcal R_z^*$,
 \[
 |\delta_D(Q,w)|=|\delta_D(Q,wc)|.
 \]
 \end{lemma}

 \begin{proof}
 For each $D\in\mathcal D_N$, first minimize
 $|\delta_D(Q,c)|$ over $c\in\mathcal R_z^*$.
 Among the words attaining this minimum, choose a shortest one and call it
 $e_D$.
 It is a concatenation of fewer than $2^N$ words from $\mathcal R_z$.
 Otherwise, the initial image $Q$ and the images after each of the first
 $2^N$ words in the concatenation would give $2^N+1$ subsets of $Q$.
 Two would be equal.
 If the same subset is reached after $i<j$ words, delete the portion of
 $e_D$ read between those two points.
 The remaining suffix then starts from the same subset as before, giving a
 shorter word with the same final image, a contradiction.

 Fix an ordering of $\mathcal D_N$, and put
 \[
 w=@1^{|z|}\prod_{D\in\mathcal D_N}e_D,
 \]
 where the product is taken in that order.
 Then $w\in\mathcal R_z^*$ and $w$ begins with $@$.

 Fix $D\in\mathcal D_N$ and write $w=a e_D b$.
 Since $\delta_D(Q,a)\subseteq Q$, and reading $b$ cannot increase the size
 of an image,
 \[
 |\delta_D(Q,w)|
 \le |\delta_D(\delta_D(Q,a),e_D)|
 \le |\delta_D(Q,e_D)|.
 \]
 Since $w\in\mathcal R_z^*$ and $e_D$ minimizes the image size, the reverse
 inequality also holds, and
 \[
 |\delta_D(Q,w)|=|\delta_D(Q,e_D)|.
 \]
 For any $c\in\mathcal R_z^*$, we have $wc\in\mathcal R_z^*$.
 By the choice of $e_D$,
 \[
 |\delta_D(Q,wc)|\ge |\delta_D(Q,e_D)|=|\delta_D(Q,w)|.
 \]
 Reading $c$ cannot increase the size of an image, so the reverse inequality
 also holds.
 Hence
 \[
 |\delta_D(Q,wc)|=|\delta_D(Q,w)|.
 \]

 Every word in $\mathcal R_z$ has length $|z|+1$, and every $e_D$ is a
 concatenation of fewer than $2^N$ such words.
 Hence
 \[
 |w|
 \le \bigl(1+2^N|\mathcal D_N|\bigr)(|z|+1)
 =N^{O(N)}.
 \]
 \end{proof}

 \subsection{Completing the construction}
 \begin{lemma}[A hard pair over the delimiter alphabet]
 \label{lem:delimiter-pair}
 There are equal-length words $u_N,v_N\in\Sigma^*$ of length $N^{O(N)}$
 such that
 \[
 \sepw(u_N,v_N)>N,
 \qquad
 u_N\in C_k,
 \qquad
 v_N\notin C_k.
 \]
 \end{lemma}

 \begin{proof}
 Take $z$ from Lemma~\ref{lem:replacement-word} and $w$ from
 Lemma~\ref{lem:stable-suffix}.
 Define the equal-length words
 \[
 g=@1^{|z|}w,
 \qquad
 b=@zw.
 \]
 Since $@1^{|z|}\in\mathcal R_z$ and $w\in\mathcal R_z^*$, we have
 $g\in\mathcal R_z^*$.

 Fix $D\in\mathcal D_N$ and put $P=\delta_D(Q,w)$.
 Since both $g$ and $b$ end in $w$, they map $P$ into itself:
 \[
 \delta_D(P,g)\subseteq P,
 \qquad
 \delta_D(P,b)\subseteq P.
 \]
 Since $g\in\mathcal R_z^*$, Lemma~\ref{lem:stable-suffix} gives
 \[
 |\delta_D(P,g)|
 =|\delta_D(Q,wg)|
 =|\delta_D(Q,w)|
 =|P|.
 \]
 Hence $\delta_D(P,g)=P$.

 We now prove the same equality for $b$.
 For any distinct states $p,q\in P$, apply
 Lemma~\ref{lem:replacement-word} to
 $\delta_D(p,@)$ and $\delta_D(q,@)$.
 It gives an $h\in\overline{G_k}$ with $|h|=|z|$ and the same action as $z$
 on both states.  Thus, for $r=p,q$,
 \[
 \delta_D(r,b)=\delta_D(r,@hw).
 \]
 Since $@h\in\mathcal R_z$, the word $@hw$ belongs to
 $\mathcal R_z^*$.
 Lemma~\ref{lem:stable-suffix} gives
 \[
 |\delta_D(P,@hw)|
 =|\delta_D(Q,w@hw)|
 =|P|.
 \]
 Thus, the distinct states $p$ and $q$ have distinct images under
 $@hw$, and hence also under $b$.
 As this holds for every distinct $p,q\in P$,
 $|\delta_D(P,b)|=|P|$.
 Together with $\delta_D(P,b)\subseteq P$, this gives
 $\delta_D(P,b)=P$,
 and both $g$ and $b$ act as permutations of $P$.

 Put
 \[
 u_N=wg^{N!},
 \qquad
 v_N=wb^{N!}.
 \]
 Since $|P|\le N$, both $g^{N!}$ and $b^{N!}$ act as the identity on $P$.
 After reading $w$, the state lies in $P$, so, for every $q\in Q$,
 \[
 \delta_D(q,u_N)=\delta_D(q,w)=\delta_D(q,v_N).
 \]
 Since $D$ was arbitrary, no DFA with at most $N$ states separates $u_N$
 and $v_N$.

 Since $w,g\in\mathcal R_z^*$, we have
 $u_N\in\mathcal R_z^*\subseteq C_k$.
 Since $w$ is also nonempty,
 $b=@zw$ contains the substring $@z@$; since
 $N!\ge1$, so does $v_N$.
 As $z\in G_k$, this gives $v_N\notin C_k$.
 Finally, $|g|=|b|$, so $u_N$ and $v_N$ have the same length.
 Since $N!\le N^N$ and $|z|,|w|=N^{O(N)}$, the length of $u_N$ and $v_N$ is $N^{O(N)}$.
 \end{proof}

 Lemma~\ref{lem:delimiter-pair} gives a pair over the four-letter alphabet
 $\Sigma$ at a construction-dependent length.
 The next two lemmas convert it to a binary pair while preserving the needed
 separation bounds, and then pad it to any larger prescribed length.

 \begin{lemma}[Reduction to binary words]
 \label{lem:palindromic-code}
 For every pair of distinct equal-length words $r,s\in\Sigma^*$, there are
 equal-length binary words $\widehat r,\widehat s$ of length $O(|r|)$ such
 that
 \[
 \sepw(r,s)\le\sepw(\widehat r,\widehat s)
 \]
 and
 \[
 \sepw(\rev{\widehat r},\rev{\widehat s})
 =O(\sepw(\rev{r},\rev{s})).
 \]
 \end{lemma}

 \begin{proof}
 As in Proposition~2 of~\cite{DemaineEtAl2011}, use the fixed-length
 palindromic binary code
 \[
 \begin{aligned}
 \gamma(@)&=000, &\gamma(1)&=010,\\
 \gamma(2)&=101, &\gamma(\#)&=111,
 \end{aligned}
 \]
 extended to words by concatenation, and put
 $\widehat r=\gamma(r)$ and $\widehat s=\gamma(s)$.
 These are binary words of length $3|r|$.
 Since every codeword is a palindrome,
 \[
 \rev{\gamma(t)}=\gamma(\rev{t})
 \]
 for every $t\in\Sigma^*$.

 Suppose that a binary DFA $E$ separates $\widehat r$ and $\widehat s$.
 For each $\sigma\in\Sigma$, define its transition to be the transition
 that $E$ makes on the three-bit word $\gamma(\sigma)$.
 With the same states, initial state, and accepting states as $E$, this
 gives a DFA over $\Sigma$ separating $r$ and $s$.
 Thus encoding cannot decrease separation.

 Conversely, let $D$ be a $q$-state DFA over $\Sigma$ that separates
 $\rev r$ and $\rev s$.
 A binary DFA reads the input in blocks of three bits.
 It records the current state of $D$ together with the one or two bits
 already read from the current block.
 When a complete block equals $\gamma(\sigma)$, it applies the
 $\sigma$-transition of $D$ and begins a new block.
 Any other three-bit block enters a rejecting sink.
 The accepting states are those reached at the end of a block when the
 corresponding state of $D$ is accepting.
 This uses $O(q)$ states and separates
 $\gamma(\rev r)=\rev{\widehat r}$ from
 $\gamma(\rev s)=\rev{\widehat s}$.
 \end{proof}

 \begin{lemma}[Padding to every target length]
 \label{lem:exact-length-padding}
 Let $x,y$ be distinct equal-length binary words.
 For every $L\ge|x|$, there are binary words $X,Y$ of length $L$ such that
 \[
 \sepw(X,Y)\ge\sepw(x,y)
 \]
 and
 \[
 \sepw(\rev X,\rev Y)
 \le\sepw(\rev x,\rev y)+1.
 \]
 \end{lemma}

 \begin{proof}
 Put $d=L-|x|$.
 If $d=0$, take $X=x$ and $Y=y$; both conclusions are immediate.
 Assume that $d>0$, and take
 \[
 X=x10^{d-1},
 \qquad
 Y=y10^{d-1}.
 \]
 Appending a common suffix cannot decrease deterministic separation:
 a DFA separating the extended words must already be in different states
 after reading $x$ and $y$, and those two states can be declared accepting and
 rejecting.
 Hence $\sepw(X,Y)\ge\sepw(x,y)$.

 Let $D$ be a smallest DFA separating $\rev x$ and $\rev y$.
 Add one new initial state that loops on $0$ and, on reading $1$, enters the
 old initial state of $D$.
 With all old transitions and accepting states unchanged, the resulting DFA
 separates
 \[
 \rev X=0^{d-1}1\rev x
 \qquad\text{and}\qquad
 \rev Y=0^{d-1}1\rev y.
 \]
 It has one more state than $D$.
 \end{proof}

 We can now prove Theorem~\ref{thm:reversal-gap}.

 \begin{proof}[Proof of Theorem~\ref{thm:reversal-gap}]
 Take $u_N,v_N$ from Lemma~\ref{lem:delimiter-pair}.
 The $O(k)$-state DFA for $\rev{C_k}$ from
 Lemma~\ref{lem:delimiter-automata} separates $\rev{u_N}$ and $\rev{v_N}$.
 Applying Lemma~\ref{lem:palindromic-code} gives equal-length binary words
 $x_N,y_N$ such that
 \[
 \sepw(x_N,y_N)>N,
 \qquad
 \sepw(\rev{x_N},\rev{y_N})=O(k).
 \]
 Since $|x_N|=N^{O(N)}$, there is an absolute constant $c_0$ such that
 \[
 |x_N|\le N^{c_0N}
 \]
 for every $N\ge2$.

 For any $L\ge N^{c_0N}$, apply
 Lemma~\ref{lem:exact-length-padding} to $x_N,y_N$.
 It gives binary words $x_{N,L},y_{N,L}$ of length $L$ with
 \[
 \sepw(x_{N,L},y_{N,L})>N
 \]
 and
 \[
 \sepw(\rev{x_{N,L}},\rev{y_{N,L}})
 =O(k)=O(\log N).
 \]
 \end{proof}

 \section{Conclusion}

 The results above give answers or refinements to the questions from
 Demaine et al.\ that motivated the paper.
 For structured deterministic pairs, their Hamming-distance method extends
 from differing positions to $d$ runs of differences, while the complexity of
 conjugate words depends on the arithmetic of the shift.
 Open Problem~2 is resolved by an unbounded difference and ratio between the
 two orientations of nondeterministic separation.
 The reversal construction addresses Ebrahimnejad's follow-up to Open Problem~1
 by giving an unbounded ratio between forward and reversed deterministic
 separation.
 Since nondeterministic separation is invariant under reversal, that same
 construction also improves the lower bound in Open Problem~3.

 Several questions remain open.
 The optimal worst-case bound in terms of the number of runs and the word length is unknown,
 as is the correct order of the separation complexity of conjugate words.
 The function $A(q)$ measuring the effect of exchanging the accepted and
 rejected words remains between quadratic and exponential.
 For Open Problem~3, determining the correct asymptotic order and reducing the
 length threshold in the reversal construction remain open.

 \section{Acknowledgments}

 We thank Jeffrey Shallit and Zachary Chase for helpful discussions regarding some of these problems.

\section{Declaration of generative AI and AI-assisted technologies in the manuscript preparation process}
During the preparation of this work the author used ChatGPT 5.6 Pro
for editing and review, and to generate initial proofs for
Theorem~\ref{thm:nsep-asymmetry} and an earlier, weaker version of
 Theorem~\ref{thm:reversal-gap}.
The author reviewed and checked all arguments,
 substantially revised and simplified the proofs,
 and takes full responsibility for the content of the article.

 \bibliographystyle{elsarticle-num}
 \bibliography{references}

\end{document}